\documentclass[12pt,a4paper]{article}
\usepackage[a4paper,margin=1in]{geometry}

\usepackage{amsmath,amsthm,amssymb,amsfonts}
\usepackage{graphicx}
\usepackage{cite}
\usepackage[table]{xcolor}

\usepackage{multirow}

\newtheorem{remark}{Remark}
\newtheorem{theorem}{Theorem}

\newtheorem{props}{Proposition}
\theoremstyle{definition}
\newtheorem{defn}{Definition}
\newtheorem{algo}{Algorithm}

\newtheorem{exmp}{Example}
\usepackage{multirow}

\begin{document}
\title{Selection of overlapping interaction through approximate decentralized fixed mode measure}
%

\author{Pradosh~Ranjan~Sahoo$^{1}$, Abhilash~Patel$^2$, Sandip Ghosh$^{3*}$, Asim~K.~Naskar$^1$\\
{\small $^1$Department of Electrical Engineering, National Institute of Technology,
Rourkela, India}\\
{\small $^2$Department of Electrical Engineering, Indian Institute of Technology, Delhi, India}\\
{\small $^3$Department of Electrical Engineering, Indian Institute of Technology (BHU), Varanasi, India}\\
{\small Corresponding* Email: sghosh@eee.iitbhu.ac.in}}

\date{}

%
%
{\let\newpage\relax\maketitle}
\begin{abstract}
This work considers the problem of selecting overlapping control structures to remove decentralized fixed modes. The selection of such controller is conventionally carried out based on minimal overlapping communications. In this paper, this selection is proposed to be through approximate decentralized fixed mode measure. Also, a framework for improving the control cost of approximate decentralized fixed modes through overlapping control is given. Application of the proposed selection is demonstrated through several numerical examples.
\end{abstract}
\maketitle


\section{Introduction}
Many control problems of modern industries and  society, e.g. electrical power systems, transportation systems and robotic systems, are dealt in the framework of control of large-scale systems
 \cite{sandell:1978, stankovic:2000, siljak:2011, siljak:2005, chang:1990}. Usually, the desirable control structure for such systems is decentralized, which consists of several
local control stations with restrictions in information flow from one control station to another.

One important issue that arises with decentralized control is the decentralized fixed modes (DFM)  The notion of DFM was first introduced by Wang and Davison \cite{Wang:1973}. DFMs are the modes of the open-loop system that are fixed with respect to any LTI decentralized controller. Different characterization methods for DFM have been discussed in literature \cite{anderson:1981,davison:1983,anderson:1982, lavaei:2007}.

Often, decentralized overlapping control that  allows information sharing among limited control stations is used to eliminate DFMs \cite{ sezer:1981,armentano:1982,unyeliouglu:1989,sojoudi:2009}. For this, finding a minimum cost feedback pattern that does not give rise to structural fixed mode is described in \cite{unyeliouglu:1989}. In \cite{sojoudi:2009}, finding optimal information structure for structurally constrained controller is considered for removing un-repeated fixed modes. Recently, characterization of DFM in the perspective of distributed control has been carried out based on overlapping estimation in \cite{sturz:2017}, where structural conditions for guaranteeing the existence of minimal communication topology are derived. Other applications, for selection of information sharing is found in multi-agent systems \cite{wei:017}, where algorithms are used to obtain stabilizing information shairing connections.

In applications, it is often encountered that modes are not exactly fixed, rather these are close to being fixed.
Such modes are called as Approximate DFMs (ADFMs) \cite{vaz:1989}. Characterization of ADFM are discussed in \cite{davison:2008, sahoo:2017}. Control of such ADFMs require large control cost and one may wish to select overlapping controller for better control of the ADFM. Note that, since ADFMs are not DFM, the existing work on overlapping controller selection for DFMs are not directly applicable.

In this paper, ADFM measure is used as the cost function for selection of overlapping interactions for DFM removal. Such a performance measure for overlapping loop selection has not been investigated so far in literature. Methodology for the same is developed and demonstrated through examples. Selection of overlapping interactions for systems with ADFMs is also discussed. For this, ADFMs are first approximated as DFMs by appropriately perturbing system matrices following the DFM characterization. Such modes are termed as Resemblant DFMs (RDFMs) and overlapping controllers are chosen for removing the RDFMs that has potential for improving ADFM measure. Application are given to demonstrate the proposed method.

\section{The system}
Consider a large-scale  system  described as:
\begin{equation}\label{e1.4}
\begin{split}
\dot{\tilde{x}}(t) &=\begin{bmatrix}\sigma&0\\0&\tilde{A}\end{bmatrix}\tilde{x}(t)+\sum\limits_{j=1}^v{\tilde{B}}_ju_j(t) \\
  y_i(t) &= {\tilde{C}}_i\tilde{x}(t)+ \sum\limits_{j=1}^v{\tilde{D}}_{ij}u_j(t),\quad i=1,\ldots,v.
	\end{split}
\end{equation}
Let $\sigma$ being a non-repeated DFM of (\ref{e1.4}). where $v$ is the number of control stations. $x(t)\in \mathbb{R}^n$ is the system state, $u_i(t)\in \mathbb{R}^{m_i}$ and $y_i(t)\in \mathbb{R}^{r_i},i=1,2,\ldots,v, $ are the input and the output, respectively, at the $i^{th}$ control station. $A$, $B_j$s and $C_i$s are real and constant matrices of appropriate dimension. Further, define $u(t)=\begin{bmatrix}u_1^T(t),\ldots,u_v^T(t)\end{bmatrix}^T$, $y(t)=\begin{bmatrix}y_1^T(t),\ldots,y_v^T(t)\end{bmatrix}^T,\,\tilde{B}=\begin{bmatrix}\tilde{B}_1,\ldots,\tilde{B}_v\end{bmatrix}$ and $\tilde{C}=\begin{bmatrix}\tilde{C}_1^T,\ldots,\tilde{C}_v^T\end{bmatrix}^T$. Let us assume that system (\ref{e1.4}), i.e. the triplet $\left(A=\begin{bmatrix}\sigma&0\\0&\tilde{A}\end{bmatrix},\tilde{B},\tilde{C}\right)$, is centrally controllable and observable. Controllability and observability together guarantees that there exists no CFMs,  but the same is not true for DFMs.

The control structure considerations on the output feedback controller $u(t)=Ky(t)$ for  system (\ref{e1.4}) is as follows. Let $\bar{v}=\{1,2,\ldots,v\}$. The controller gain matrix $K$ has block-entry $K_{ij}, i,j\in\bar{v}$, if $y_j(t)$ contributes in constructing $u_i(t)$. Note that, any sort of structural constraint on $K$ will lead to a decentralized controller. A perfect decentralized controller contains only strict interactions at the respective control stations, i.e. $K=K_D$, $K_D=blockdiag(K_{11},\ldots,K_{vv})$. The off-diagonal entries in $K$ contributes to overlapping interactions and will be contained in $K_e$. An overlapping decentralized controller will be denoted as $K=K_D\cup K_e$.
\section{DFM and its removal}
\subsection{Decentralized Fixed Mode}
 A significant problem associated with the decentralized control of system (\ref{e1.4}) is the occurrence  of Decentralized Fixed Mode (DFM) \cite{Wang:1973} defined as below.
\begin{defn} System (\ref{e1.4}) is said to have a DFM at $\sigma\in\mathbb{C}$, if $\sigma\in sp(A)$  and it is so for any LTI decentralized feedback controller. A DFM satisfies the below condition
\begin{equation}\label{e1.3}
\sigma\in sp(A+\tilde{B}K_D(I-\tilde{D}K_D)^{-1}\tilde{C}),\qquad \forall K_D
\end{equation}
where $|(I-\tilde{D}K_D)|\neq0$.
\end{defn}

\begin{theorem}[\cite{lavaei:2007}]\label{th1} If there exists a permutation of $\{1,\ldots,v\}$ denoted by distinct integers $i_1,i_2,\ldots,i_v$ and $w\in\{1,2,\ldots,v-1\}$ such that $M^{\gamma,\eta}=\tilde{C}_{\gamma}\begin{bmatrix}0&0\\0&(\tilde{A}-\sigma I_{n-1})^{-1}\end{bmatrix}\tilde{B}_{\eta}-\tilde{D}_{\gamma,\eta}$ is a zero matrix $\forall \eta\in\{i_1,i_2,\ldots,i_w\}$, $\forall \gamma\in\{i_{w+1},i_{w+2},\ldots,i_v\}$ and the first row of the matrices $\tilde{B}_1,\tilde{B}_2,\ldots,\tilde{B}_{i_w}$ and the
first column of the matrices $\tilde{C}_{i_{w+1}},\ldots,\tilde{C}_{i_v}$ are all zero, then the mode $\sigma$ is a  DFM of system (\ref{e1.4}).
\end{theorem}
In view of the above, let us define
\begin{equation}\label{e1.5}
M=\begin{bmatrix}\tilde{C}_1\\\vdots\\ \tilde{C}_v\end{bmatrix}\begin{bmatrix}0&0\\0&(\tilde{A}-\sigma I)^{-1}\end{bmatrix}\begin{bmatrix}\tilde{B}_1^T\\ \vdots\\ \tilde{B}_v^T\end{bmatrix}^T-
\begin{bmatrix}\tilde{D}_{11}&\ldots&\tilde{D}_{vv}\\ \vdots& \ddots&\vdots \\ \tilde{D}_{v1}&\ldots&\tilde{D}_{vv}\end{bmatrix}
\end{equation}

Note that, $M^{\gamma,\eta}\in\mathbb{C}^{r_{\gamma}\times m_{\eta}}$ can be obtained from $M$ by concatenating $w$ rows and $(v-w)$ columns corresponding to the sequence $i_1,i_2,\ldots,i_v$.

\subsection{Overlapping controller for DFM removal}

Let the set $K_e=\{K_{p_1q_1},K_{p_2q_2},\ldots,K_{p_{\alpha} q_{\alpha}}\}$, where $p_i,q_i\in\{1,2,\ldots,v\}$, $i=1,\ldots,\alpha$, are the  desired overlapping control interactions for removing the DFM, i.e. $\sigma$ is not a DFM for the feedback gain $K=K_D\cup K_e $.
%

Further, let the system contains multiple non-repeating DFMs $\sigma_1,\sigma_2,\ldots,\sigma_{q}$ that are to be removed. The overlapping controller for removing each of the fixed modes $\sigma_i, i=1,\ldots,q$, can be obtained by the procedure discussed in \cite{sojoudi:2009} and denoted as $K_e^{1},K_e^{2},\ldots,K_e^{q}$. Then the desired overlapping controller may be chosen as
\begin{equation}\label{c1.10}
K_e\supseteq {K_e^{1}\cup K_e^{2},\ldots,\cup,K_e^{q}}.
\end{equation}
Note that, the above selection only ensures removal of the DFM. Out of several such possible selections, one may require to choose the most economical one. One approach for this would be to choose the minimal overlapping connections and then finding best one out of them. Such a method has been adopted in \cite{sojoudi:2009}. However, the choice of minimal overlapping connection is not central of this work, rather choosing the best one from the set that may include even the non-minimal ones is the objective here. It is worth noting in this regard that minimal connection may not ensure minimal cost performance.


\section{ADFM and improving its measure}
For decentralized control of system (\ref{e1.4}), it may be the case that the system does not have an explicit DFM, yet some of the modes behave similar to the DFMs, e.g. consumption of large control effort for significant relocation of a particular mode. Such modes are referred to as ADFMs \cite{davison:2008} as defined below.
\begin{defn} System (\ref{e1.4}) is said to have an ADFM at $\sigma\in \mathbb{C}$, if $\sigma\in sp(A)$  and  $\forall K_D$,
 \begin{equation}\label{e1.2}
 |\hat{\sigma}-\sigma|\leq\epsilon\,\,  \mbox{\,for\,} \hat{\sigma}\in sp(A+\tilde{B}K_D(I-\tilde{D}K_D)^{-1}\tilde{C}) 
 \end{equation}
 where $\epsilon$ is a small positive number.
\end{defn}
\subsection{ADFM measure}
One way to characterize ADFMs is through the condition number measure \cite{davison:2008}, defined as:
\begin{align}\label{e1.11a}
& \mathcal{D}(\sigma):=min\{\{cond(W_i(\sigma)),i=1,...,v \},\nonumber\\
& \{cond(W_{i,j}(\sigma)),i=1,...,v-1,j=i+1,...,v\},\nonumber \\
&\{cond(W_{i,j,k}(\sigma )),i=1,...,v-2,j=i+1,...,v-1,\nonumber \\  &k=j+1,...,v\}
...\{cond(W_{1,2,....,v})\}\},
 \end{align}
where $cond(\cdot)$ represents the condition number of $(\cdot)$ and $W_{1,\ldots,v}(\cdot)$ are matrices corresponding to the subsystems that are used to characterize DFM in terms of transmission
zeros \cite{davison:2008} as represented below:
\begin{align}\label{e1.8i}
&W_i(\sigma)=\begin{bmatrix}A-\sigma I& \tilde{B}_i\\\tilde{C}_i&0\end{bmatrix}, W_{i,j}(\sigma)=\begin{bmatrix}A-\sigma I& \tilde{B}_i&\tilde{B}_j\\\tilde{C}_i&0& \tilde{D}_{ij}\\\tilde{C}_j&\tilde{D}_{ji}&0\end{bmatrix},\nonumber \\
&W_{i,j,k}(\sigma)=\begin{bmatrix}A-\sigma I& \tilde{B}_i&\tilde{B}_j&\tilde{B}_k\\\tilde{C}_i&0&\tilde{D}_{ij}&\tilde{D}_{ik}\\\tilde{C}_j&\tilde{D}_{ji}&0&\tilde{D}_{jk}\\\tilde{C}_k&\tilde{D}_{ki}&\tilde{D}_{kj}&0\end{bmatrix},\cdots\nonumber\\
&W_{1,2,...,v}(\sigma)=\begin{bmatrix}A-\sigma I& \tilde{B}_1&\tilde{B}_2&\ldots&\tilde{B}_v\\\tilde{C}_1&0&\tilde{D}_{12}&\ldots&\tilde{D}_{1v}\\C_2&\tilde{D}_{21}&0&\ldots&\tilde{D}_{2v}\\ \vdots&\vdots&\vdots&\ddots&\vdots\\ \tilde{C}_v&\tilde{D}_{v1}&\tilde{D}_{v2}&\ldots&0\end{bmatrix}
\end{align}
Larger the $\mathcal{D}(\sigma)$, $\sigma$ is that closer to being a DFM. For an explicit DFM, $\mathcal{D}(\sigma)$ is infinite.

For the system (\ref{e1.4}), condition number measure corresponding to a DFM $\sigma$ is $\infty$. It can be seen from (\ref{e1.8i}) as follows. Note that, there exists a permutation for which the condition of Theorem 1 satisfies. Also, the particular permuted partition is a case of  (\ref{e1.8i}) for which the corresponding $W_{\dots}(\sigma)$ will have its first row and first column as zero (rank deficient), as per Theorem 1. Then following (\ref{e1.11a}), $\mathcal{D}(\sigma)=\infty$.


In case of $\sigma$ is an ADFM, the first row of the matrices $\tilde{B}_1,\tilde{B}_2,\ldots,\tilde{B}_{i_w}$ and the
first column of the matrices $\tilde{C}_{i_{w+1}},\ldots,\tilde{C}_{i_v}$ are not exactly zero but nearly close to zero. The corresponding singular value of $W_{\dots}(\sigma)$ is small and $\mathcal{D}(\sigma)$ is very large value leading to identification of an ADFM. The following example demonstrates mode characterization through ADFM measure.

\begin{exmp}\label{ex2}
Consider system (\ref{e1.4}) similar to the one in \cite{sojoudi:2009}:
\begin{eqnarray*}\label{e1.10}
&&A=\begin{bmatrix}1&0&0&0\\0&2&0&0\\0&0&3&0\\0&0&0&4\end{bmatrix},\quad B=\begin{bmatrix}3&0&0.005&0\\4&2&7&0.002\\0&0&9&8\\1&6&-5&7\end{bmatrix}, \\
&&C=\begin{bmatrix}0&2&4&3\\0.0066&-6&0&8\\0.0010&4&0.0005&-9\\5&1&0.0001&7\end{bmatrix},\\
&&D=\begin{bmatrix}-5&10&227&23\\32&60&-3&56/3\\-25&-62&43&-21\\-4.5&40&16&7\end{bmatrix}.  
\end{eqnarray*}

The ADFM measures of the system eigenvalues are computed as: $\mathcal{D}(1)=1.63\times10^5$, $\mathcal{D}(2)=13.36$, $\mathcal{D}(3)=0.25\times 10^5$, $\mathcal{D}(4)=10.07$. It is clear that the modes $\sigma=1,3$ are ADFMs.
\end{exmp}

To this end, one can use overlapping  controller discussed in the previous section for improving the ADFM measure. For this, the first step is to choose effective overlapping connections. A selection approach would be to approximate such ADFMs to DFMs and then the overlapping connections are selected to remove the approximated DFMs.

\subsection{Resemblant DFMs and overlapping loop selection}
 It may be noted that, for systems with ADFMs, a small perturbation in the system matrices, e.g. setting the small elements of $A,\tilde{B}$ and $\tilde{C}$ matrices to zero makes the ADFM to be a DFM. However it is not general in practice and requires a systematic procedure. First, let us define the following.
 \begin{defn}[Resemblant DFM (RDFM)] If $\tilde{B}$ and $\tilde{C}$ matrices in (\ref{e1.4}) are perturbed in such a way that an ADFM $\sigma$ becomes a DFM, then $\sigma$ is called as RDFM.
 \end{defn}

In light of the above, the following is proposed.
\textbf{Notation:} for any $i,j\in \bar{v}$ assume that $m_i=r_i=1$.
Denote the entries of $\tilde{B}_i$ is represented as $\tilde{B}_i^{\mu1}$ and entries of $\tilde{C}_i$ as $\tilde{C}_i^{\mu1}$ where($\mu_1\in{1,\ldots,n}$) and entries of $\tilde{D}$ is denoted as $\tilde{D}_{ij}$. There exist a permutation of $\{1,\ldots,v\}$, say $i_1,\ldots,i_v$ and a partition index $w\in\{1,\ldots,v-1\}$ with $\eta\in\{i_1,i_2,\ldots,i_w\}$, $\gamma\in\{{i_{w+1}},i_{w+2},\ldots,i_v\}$,
\begin{props} There exists a scalar $\varepsilon>0$ for which $|\tilde{B}_{\eta}^1|,|\tilde{C}_{\gamma}^1|$ and $|M_{i,j}|\leq \varepsilon$ can be approximated to zero so that an ADFM $\sigma$ becomes a RDFM.
\end{props}
\begin{proof}
Define $\phi_i=\begin{bmatrix}\tilde{B}_i^2&\tilde{B}_i^3\ldots\tilde{B}_i^n\end{bmatrix}^T$ and $\Psi_i=\begin{bmatrix}\tilde{C}_i^2&\tilde{C}_i^3\ldots\tilde{C}_i^n\end{bmatrix}$ where $i\in\bar{v}$.
If $\sigma$ is a DFM of (\ref{e1.4})then rank of given matrix below is less than $n$
\begin{equation}\label{e1.24}
\begin{bmatrix}\begin{bmatrix}0&0\\0&\tilde{A}-\sigma I_{n-1}\end{bmatrix}&\tilde{B}_{i_1}&\ldots&\tilde{B}_{i_w}\\
\tilde{C}_{i_{w+1}}&\tilde{D}_{{i_{w+1}i_1}}&\ldots&\tilde{D}_{{i_{w+1}i_w}}\\
\vdots&\vdots&\ddots&\vdots\\
\tilde{C}_{i_{v}}&\tilde{D}_{{i_{v}i_1}}&\ldots&\tilde{D}_{{i_{v}i_w}}
\end{bmatrix}
\end{equation}
The multiplicity of $\sigma$ is one so the rank of the matrix $\begin{bmatrix}0&0\\0&A-\sigma I_{n-1}\end{bmatrix}$ is $n-1$.
The rank of the matrix (\ref{e1.24})is less than $n$

\begin{itemize}
  \item If all elements of the first row or the first column of given matrix (\ref{e1.24})is zero, i.e.$|\tilde{B}_{\eta}^1|,|\tilde{C}_{\gamma}^1|=0$ for any $\eta$ and $\gamma$.
  \item and if $rank\begin{bmatrix}\tilde{A}-\sigma I_{n-1}&\phi_{\eta}\\ \Psi_{\gamma}&D_{\gamma\eta}\end{bmatrix}<n$ i.e To satisfy the rank condition it can be concluded that determinant of above matrix is zero. So the norm of the matrix $D_{\gamma\eta}-\Psi_{\gamma}(\tilde{A}\sigma I_{n-1})^{-1}\phi_{\eta}$ is zero.
\end{itemize}
From the definition of $M$ matrix it is clear that $M^{\gamma,\eta}=\Psi_{\gamma}(\tilde{A}-\sigma I_{n-1})^{-1}\phi_{\eta}-D_{\gamma\eta}$. So in order to convert ADFM to RDFM the entries of $|\tilde{B}_{\eta}^1|,|\tilde{C}_{\gamma}^1|$ and $|M_{i,j}|$ which are close to zero can be converted to zero.
\end{proof}
The following algorithm can be used to obtain a RDFM. Note that, for oscillatory modes, the elements of the constituent matrices would be complex. Hence, $|B|$, $|C|$ and $|M|$ are used in the above as well as in the below algorithm.
\begin{algo}[Obtaining RDFM]
   \item[\textbf{Step 1:}] Assume the ADFM $\sigma$ has been identified through ADFM measure. Transform the system into (\ref{e1.4}).
	\item[\textbf{Step 2:}] Select a permutation of $\{1,\ldots,v\}$, say $i_1,\ldots,i_v$, and a partition index $w\in\{1,2,\ldots,v-1\}$ with $\eta\in\{i_1,i_2,\ldots,i_w\}$, $\gamma\in\{{i_{w+1}},i_{w+2},\ldots,i_v\}$, so that (i) norm of the corresponding $M^{\gamma,\eta}$ and (ii) the entries of $|\tilde{B}_{\eta}^1|$ and $|\tilde{C}_{\gamma}^1|$, are close to zero.
  \item[\textbf{Step 3:}]  Replace the entries of $M^{\gamma,\eta}$, $|\tilde{B}_{\eta}^1|$ and $|\tilde{C}_{\gamma}^1|$ to zero. Then $\sigma$ is a RDFM of the system. 
 \end{algo}
 For Example \ref{ex2}, $M$ matrix for the ADFM $\sigma=1$ is:
 \begin{equation}M=\begin{bmatrix}
 14  &       0 &        0 &   0.004\\
  -53.333 & -56 & -52.333  & -0.012\\
   38 &  52 &   0.002 &   0.010\\
   10.833 & -24 & -20.666 &   9.336\end{bmatrix}
 \end{equation}
It can be worked out that $\sigma=1,3$ are not DFM because there is no $\gamma$ and $\eta$ so that the first row of the matrices $\tilde{B}_{\eta}$ (following Theorem 1) and the first column of the matrices $\tilde{C}_{\gamma}$ are all zero and corresponding $M^{\gamma,\eta}$ are also zero.

For example \ref{ex2}, setting $\varepsilon=0.015$ solves the purpose. Replacing the small elements $(0.004, -0.012, 0.002$ and $0.01)$ of $M$ with zero, and the elements of first row of $\tilde{B}$ and first column of $\tilde{C}$ matrices with values $0.005, 0.007$ and $0.001$ are made zero for $\sigma=1$ to be a RDFM. For $\sigma=1$ to be RDFM, $\eta=\{2,3,4\}$ and $\gamma=\{1\}$ is one set and $\eta=\{4\}$ and $\gamma=\{1,2,3\}$ is another set for which $M^{\gamma,\eta}=[M_{12}\,M_{13}\,M_{14}]=\begin{bmatrix}0&0&0.004\end{bmatrix}$ and $[M_{14}\,M_{24}\,M_{34}]^T=\begin{bmatrix}0.004&-0.012&0.01\end{bmatrix}^T$ are made zero, respectively. Similarly $\sigma=3$ can be made to be a RDFM.

\begin{remark}
A convenient way to implement Algorithm 1 is to permute the $\tilde{B}_i$ and $\tilde{C}_i$ matrices so that either the first row elements of $|\tilde{B}_i|$s are in ascending order or the first column elements of $|\tilde{C}_i|$s are in descending order so that selection of $\epsilon$ and thereby $w$ become easier.
\end{remark}

Next, overlapping control structure is to be chosen for removing the RDFMs. Following Theorem 3 in \cite{sojoudi:2009}, the following minimal overlapping sets for $\sigma=1$ can be chosen.
\begin{equation}\label{e1.11}
K_e^{1,1}=\{K_{14}\},K_e^{1,2}=\{K_{12},K_{34}\},K_e^{1,3}=\{K_{13},K_{24}\}
\end{equation}
Similarly, the overlapping sets for $\sigma=3$ can be obtained as:
\begin{equation}\label{e1.12}
K_e^{3,1}=\{K_{31}\},\quad K_e^{3,2}=\{K_{41}\}
\end{equation}
For removing both the RDFMs, the overlapping controller is to be chosen from the set
\begin{equation}\label{e1.11}
K_e=\left\{(K_e^{1,i}\cup K_e^{3,i}), i\in\{1,2,3\}, j=\{1,2\}\right\}
\end{equation}
Given the above choices of overlapping interactions, the economical one is chosen based on the ADFM measures corresponding to different overlapping interactions.
\subsection{ADFM measure for selection of overlapping interactions}
ADFM measure is computed for a particular feedback structure using (\ref{e1.11a}). This procedure is repeated for all the possible overlapping interactions and the combination with smallest ADFM measure is taken as the least costly interaction set. Note that, selection of overlapping interactions for DFM and RDFM are the same.

However, for computing (\ref{e1.11a}), one requires diagonal decentralized control structure, which is not the case for overlapping control. For this purpose, the overlapping controller  $K=K_d\cup K_e$ is transformed to a perfect decentralized structure \cite{davison:2008} as discussed below.

Consider system (\ref{e1.4}) with $D=0$ and the below overlapping controller.
 \begin{equation*}K=\begin{bmatrix}K_{11}&0&0&K_{14}\\0&K_{22}&0&0\\0&0&K_{33}&0\\K_{41}&0&0&K_{44}\end{bmatrix}.\end{equation*}
Then $A+BKC=A+\sum \limits_{i=1}^v\sum \limits_{j=1}^v  B_iK_{ij}C_j=A+\bar{B}\bar{K}\bar{C}$, where $\bar{K}$ is diagonal containing all $K_{ij}$, and
\begin{eqnarray*}
\bar{C}&=&\begin{bmatrix}C_1&C_4&C_2&C_3&C_1&C_4\end{bmatrix}^T, \\
\bar{B}&=&\begin{bmatrix}B_1&B_1&B_2&B_3&B_4&B_4\end{bmatrix}
\end{eqnarray*}
Then, the ADFM measure of the triplet $(\bar{C},A,\bar{B})$ can be computed using (\ref{e1.11a}) that corresponds to the original triplet $({C},A,{B})$ but with the overlapping controller.

The computed ADFM measures for the overlapping interactions for Example \ref{ex2} are shown in Table \ref{table1}. It can be seen that the measures are considerably improved for the overlapping controller as compared to the perfect decentralized controller. Also, a suitable controller choice is $ K_e=\{k_{14},k_{41}\}$ as it improves ADFM measure for both the modes considerably as compared to other interactions.
%


\begin{table}[t]
\centering
\caption{ADFM measures for different controller interaction}\label{table1}
\begin{tabular}{|c|c|c|}
  \hline
  \multirow{2}{*}{$K_e$} & \multicolumn{2}{l|}{ADFM Measure} \\ \cline{2-3}
                 &$\sigma_1=1$ & $\sigma_3=3$ \\ \hline
  $K_{14},K_{31}$ & $15.86$& $22.88$ \\ \hline
  $K_{14},K_{41}$ & \textbf{$15.86$} &\textbf{$18.26$} \\ \hline
  $K_{13},K_{24},K_{41}$ & $124.49$& $18.26$ \\ \hline
  $K_{12},K_{34},K_{31}$ & $20.84$ & $22.88$  \\
  \hline
\end{tabular}
\end{table}
\section{Conclusion}
Use of ADFM measure has been advocated for selecting the overlapping interactions for eliminating DFMs. Strating from characterization of DFMs to calculating ADFM measure has been listed and a procedure for overlapping loop selection is presented. The ADFM measure has also been used  in developing a framework for selection of overlapping interactions for systems with ADFMs. For the purpose, the notion of Resemblant DFM (RDFM) has been introduced. The efficacy of the proposed method has also been demonstrated through its applications in wide-area loop selection in power systems.

\bibliographystyle{IEEEtran}
\bibliography{poweref}

\end{document}